\documentclass[conference,10pt]{IEEEtran}
\IEEEoverridecommandlockouts
\usepackage{cite}
\usepackage{amsmath,amssymb,amsfonts}
\usepackage{graphicx}
\usepackage{textcomp}
\usepackage{xcolor}
\usepackage{amsthm}
\usepackage{graphics}
\usepackage{multicol}
\usepackage{lipsum}
\usepackage{color}
\usepackage{mathrsfs}
\usepackage{mathtools}
\usepackage{amsbsy}
\usepackage[colorlinks=true,bookmarks=false,citecolor=blue,urlcolor=blue]{hyperref}
\usepackage{xcolor}
\usepackage{mathrsfs}
\usepackage{setspace}
\usepackage{float}
\usepackage{graphicx}
\usepackage{pstool}
\usepackage{cite}
\usepackage{lettrine}
\usepackage[normalem]{ulem}
\usepackage{latexsym}
\usepackage{algpseudocode}
\usepackage{algorithm,algpseudocode}
\usepackage{algorithmicx}
\usepackage{multirow}
\usepackage{wasysym}
\usepackage{cite}
\usepackage{mathrsfs}
\usepackage{amsmath,cite,amsfonts,amssymb,color}

\ifCLASSOPTIONcompsoc
\usepackage[caption=false,font=normalsize,labelfon
t=sf,textfont=sf]{subfig}
\else
\usepackage[caption=false,font=footnotesize]{subfig}
\fi
\allowdisplaybreaks

\newtheorem{lemma}{Lemma}
\newcommand{\imj}{\mathrm{j}}

\graphicspath{{figures/}}
\def\BibTeX{{\rm B\kern-.05em{\sc i\kern-.025em b}\kern-.08em
    T\kern-.1667em\lower.7ex\hbox{E}\kern-.125emX}}
\begin{document}

\title{Structured Sensing Matrix Design for In-sector Compressed mmWave Channel Estimation
}

\author{\IEEEauthorblockN{Hamed Masoumi$^{\ast }$, Nitin Jonathan Myers$^{\ast }$, Geert Leus$^{\dagger}$, Sander Wahls$^{\ast }$ and Michel Verhaegen$^{\ast }$}
\IEEEauthorblockA{$^{\ast }$Delft Center for Systems and Control, Delft University of Technology, The Netherlands\\
$^{\dagger}$Department of Microelectronics, Delft University of Technology, The Netherlands\\
Email: \{H.Masoumi, N.J.Myers, G.J.T.Leus, S.Wahls, M.Verhaegen\}@tudelft.nl}
}

\maketitle
\begin{abstract}
Fast millimeter wave (mmWave) channel estimation techniques based on compressed sensing (CS) suffer from low signal-to-noise ratio (SNR) in the channel measurements, due to the use of wide beams. To address this problem, we develop an in-sector CS-based mmWave channel estimation technique that focuses energy on a sector in the angle domain. Specifically, we construct a new class of structured CS matrices to estimate the channel within the sector of interest. To this end, we first determine an optimal sampling pattern when the number of measurements is equal to the sector dimension and then use its subsampled version in the sub-Nyquist regime. Our approach results in low aliasing artifacts in the sector of interest and better channel estimates than benchmark algorithms.
\end{abstract}


\section{Introduction}
Millimeter wave systems enable Gbps data rates by employing large antenna arrays and leveraging beamforming. The overhead for channel estimation, however, increases with large arrays. This calls for fast and computationally efficient techniques for high-dimensional channel estimation. In this regard, compressed sensing (CS) \cite{candes2008introduction} has been used for sub-Nyquist mmWave channel estimation \cite{alkhateeb2015compressed,myers2019falp,tsai2018structured,wang2021jittering}, as CS can leverage the sparse angle domain representation of mmWave channels \cite{heath2016overview}.

\par Standard CS techniques usually employ quasi-omnidirectional beams to acquire channel measurements \cite{alkhateeb2015compressed,myers2019falp,ali2017millimeter}. The received SNR with such beams, however, is poor. One way to overcome the low SNR issue is to focus energy in a certain band of spatial angles defined as a sector \cite{tsai2018structured,chen2020convolutional,wang2021jittering}. We call CS methods that use such beams as \textit{in-sector CS}. The key challenge with in-sector CS is to construct an ensemble of beams that all focus within the sector of interest. Such an ensemble can be constructed by circulantly shifting a beam focusing within a sector \cite{myers2019falp}. The resulting CS technique, called convolutional CS \cite{li2012convolutional}, acquires measurements by projecting the channel onto different circulant shifts of a beam training vector. 

\par In this paper, we explain that the standard approach of using random circulant shifts in convolutional CS results in ``uniform'' aliasing artifacts in the angle domain. Such an aliasing profile is not suitable for in-sector CS, where the goal is to estimate the sparse channel over a specific support set. To this end, we design the circulant shifts so that the aliasing artifacts are low within the sector of interest and are high outside the sector.
To aid our design, we first determine the optimal set of circulant shifts when the number of CS measurements is equal to the dimension of the sector. Finally, we use the subsampled version of the optimal set of circulant shifts for channel estimation in the sub-Nyquist sampling regime. Our design can be integrated into the IEEE 802.11.ad/ay standard \cite{nitsche2014ieee}, and it achieves lower in-sector channel reconstruction error than comparable benchmarks.
 \par We now discuss prior work related to in-sector CS. In \cite{tsai2018structured}, randomly selected DFT columns were used to modulate a spread sequence and generate different in-sector beams. Such an approach, however, degrades the received SNR in the sector of interest. In \cite{wang2021jittering}, a contiguous band of directions is illuminated by sub-arrays. Then, the beam at each sub-array is phase modulated to construct an ensemble of in-sector CS beams. The greedy method in \cite{ali2017millimeter} selects CS beams from a large set of random codes, based on their capability to concentrate energy within a band of interest. The greedy CS matrices in \cite{wang2021jittering} and \cite{ali2017millimeter}, however, still result in substantial aliasing artifacts within the sector of interest. Finally, \cite{chen2020convolutional} performs linear convolution over the output of a fully digital array, and then decimates the output of each filter to achieve parallel processing.  In this paper, we realize subsampled circular convolution using the analog beamforming architecture, and propose a randomized construction of CS matrices in the sub-Nyquist regime. While the approach in \cite{chen2020convolutional} reduces the channel dimension to achieve steady state linear convolution, our circular convolution-based method preserves the channel dimension. The downside of our CS matrix design is the assumption that the sources are on-grid, unlike \cite{chen2020convolutional} where the sources can be off-grid. For the special case when the number of CS measurements is equal to the sector dimension, the CS matrix proposed in this paper is same as the deterministic construction in \cite{chen2020convolutional}. 

\emph{Notation}: $a$, $\mathbf{a}$ and $\mathbf{A}$ denote a scalar, vector, and a matrix. The indexing of vectors and matrices begins at $0$. $\mathbf{a}_i$ is the $i^{\mathrm{th}}$ column of $\mathbf{A}$. 
$(\cdot)^{T}$ and $(\cdot)^{\ast}$ denote the transpose and conjugate-transpose operators. $\mathrm{Diag}(\mathbf{a})$ is a diagonal matrix with $\mathbf{a}$ on the diagonal. The matrix $|\mathbf{A}|$ contains the element-wise magnitudes of $\mathbf{A}$. $\circledast$ and $\odot$ denote circular convolution and Hadamard product. The inner product $\langle\mathbf{a},\mathbf{b}\rangle=\mathbf{b}^{\ast}\mathbf{a}$. $[N]$ denotes the set $\{0,1,\cdots,N-1\}$. $\mathcal{CN}(0,\sigma^{2})$ is the zero-mean complex Gaussian distribution with variance $\sigma^{2}$. $\mathsf{j}=\sqrt{-1}$.
\section{Channel and system model}\label{sec2sysmodel}
We consider a point-to-point mmWave system, shown in Fig.~\ref{fig:sys}, where the transmitter (TX) has a uniform linear array (ULA) with $N$ omni-directional radiating elements. The TX uses an analog beamforming architecture that allows both amplitude and phase control at each antenna. The receiver (RX) has a single omni-directional antenna. The TX transmits training signals and the RX uses the corresponding received measurements to estimate the channel. 
\subsection{Channel model}
 We use $\mathbf{h}\in\mathbb{C}^{N\times 1}$ to denote the geometric narrowband mmWave channel between the TX and the RX. We assume $K$ propagation rays from the TX to the RX \cite{heath2016overview}. Each ray $k\in \{1,2,\cdots,K\}$ has an angle of departure (AoD) $\theta_{k}$ and a complex gain $\alpha_{k}$.
 For a half-wavelength spaced antenna array at the TX, we use $\mathbf{a}(N,\theta_{k})\in\mathbb{C}^{N\times 1}$ to denote the array steering vector corresponding to the $k^{\mathrm{th}}$ AoD at the TX, i.e.,
\begin{align}\label{eqn:ArrayResponse}
    \mathbf{a}(N,\theta_{k}) = [1, e^{\mathsf{j}\pi \sin(\theta_{k})}, e^{\mathsf{j}2\pi \sin(\theta_{k})}, ..., e^{\mathsf{j}(N-1)\pi\sin(\theta_{k})}]^{T}.
\end{align}
The channel between the TX and the RX, i.e., $\mathbf{h}$, is then
\begin{align}\label{eqn:PhyChannel}
    \mathbf{h} = \sum\limits_{k=1}^{K}\alpha_{k}\mathbf{a}(N,\theta_{k}).
\end{align}
The average power of the channel is $\mathbb{E}[\mathbf{h}^{\ast}\mathbf{h}]=N$.
\begin{figure}[!t]
\centering
\includegraphics[scale=0.86]{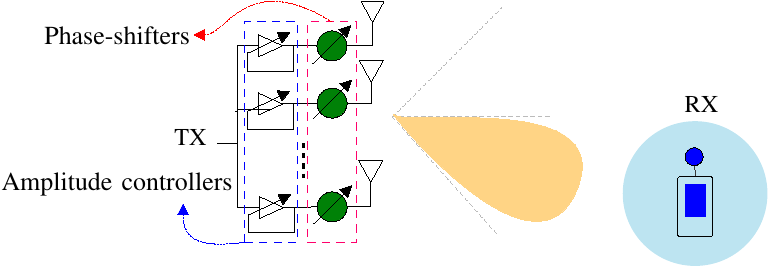}\label{fig:SNRimp_SW}
\vspace{-1.5mm}
\caption{\small An mmWave system with a ULA at the TX and a single antenna RX. We focus on spatial channel estimation within a sector.
\normalsize}\label{fig:sys}
\end{figure}
\par MmWave channels are approximately sparse in the angle domain due to high scattering at mmWave wavelengths \cite{heath2016overview}. Thus, the beamspace representation is convenient to exploit the sparsity of mmWave channels. Due to the use of a ULA at the TX, the DFT dictionary can be used to represent the beamspace \cite{heath2016overview}. We define $\mathbf{U}_{N}$ as the standard unitary $N \times N$ DFT matrix. The beamspace channel, denoted by $\mathbf{g}$, is then
\begin{align}\label{eqn:beamspace}
    \mathbf{g} = \mathbf{U}_{N}^{\ast}\mathbf{h}.
\end{align}
To aid the construction of our CS matrix, we assume that $\mathbf{g}$ is exactly sparse. Such a scenario occurs when all the AoDs in the channel coincide with any of the $N$ angles associated with the columns of the DFT. Our simulations, however, use realistic channels where the angles can be ``off-grid'' and $\mathbf{g}$ is only approximately sparse. Accounting for these off-grid effects within the design of in-sector beam training vectors is interesting for future research.
\subsection{System model and motivation for in-sector CS}
The channel measurements at the RX are obtained by applying $M$ distinct unit-norm beam training vectors at the TX. We denote the TX beam training vector used for the $m^{\mathrm{th}}$ channel measurement by $\mathbf{f}_m\in\mathbb{C}^{N\times 1}$. Further, we use $n[m]$ to denote the white Gaussian noise at the RX distributed as $\mathcal{CN}(0,\sigma^{2})$.
Assuming a block-fading channel, the received channel measurement with the $m^{\mathrm{th}}$ beam training vector is
\begin{align}\label{eqn:Measurement}
    y[m] = \langle\mathbf{h},\mathbf{f}_m\rangle + n[m],\ \ \forall m\in\{0,1,2,..,M-1\}.
\end{align}
To represent the $M$ measurements in compact form, we construct $\mathbf{y} = [y[0],y[1], ..., y[M-1]]^{T}$. Defining $\mathbf{F} = [\mathbf{f}_0,\mathbf{f}_1, ..., \mathbf{f}_{M-1}]^{\ast}$ and $\mathbf{n} = [n[0], n[1], ..., n[M-1]]^{T}$, we can then rewrite \eqref{eqn:Measurement} as
\begin{equation}\label{eqn:RXM}
    \mathbf{y} = \mathbf{F}\mathbf{h} + \mathbf{n}.
\end{equation}
The objective of CS-based channel estimation is to estimate $\mathbf{h}$ in \eqref{eqn:RXM} from $M\ll N$ channel measurements. 
 \par The goal of in-sector CS is to estimate only the entries of the sparse angle domain channel that correspond to the sector of interest, i.e., a sub-vector of $\mathbf{g}$. In-sector CS integrates well within the IEEE 802.11ad standard. To see this, we briefly describe beam search in the standard. In this procedure, the beam training interval (BTI) comprises \textit{i) sector level sweep (SLS)}, and \textit{ii) beam refinement protocol (BRP)} \cite{nitsche2014ieee}. During SLS, the TX applies different sectored beam patterns and determines the sector that achieves the highest received SNR. In BRP, the sector selected in SLS is further refined by trying different beam training vectors at the TX or the RX \cite{nitsche2014ieee}. A naive approach for BRP is one that exhaustively scans all the directional beams within the selected sector. Such an approach, however, results in a substantial training overhead for large sector dimensions. To avoid this overhead, we develop an in-sector CS solution that compressively estimates $\mathbf{g}$ at the indices corresponding to the sector of interest.
 \par Our CS-based solution acquires distinct spatial channel measurements by circulantly shifting a wide beam that just covers the sector of interest. As circulant shifts do not change the magnitude of the DFT, all the beams in our method concentrate energy within the same sector. The main focus of this work is on constructing the sequence of circulant shifts in convolutional CS for efficient in-sector channel reconstruction.
\section{Proposed design for in-sector CS}\label{sec4}
\subsection{Preliminaries on convolutional CS}
\par In convolutional CS, the TX applies circulant shifts of a base transmit beam training vector $\mathbf{f}_{\mathrm{b}}$ for the RX to acquire channel measurements \cite{li2012convolutional}. Specifically, $\mathbf{f}_{m}$ is obtained by circulantly shifting $\mathbf{f}_{\mathrm{b}}$ by $c[m]$ units. Here, $c[m]$ is an integer in $\{0,1,...,N-1\}$. The set of $M$ distinct circulant shifts used at the TX is denoted by $\Omega=\{c[0],c[1],\cdots,c[M-1]\}$. We define $\mathbf{f}_{\mathrm{b}}^{\mathrm{FC}}$ as the flipped and conjugated version of $\mathbf{f}_{\mathrm{b}}$, i.e., $f_{\mathrm{b}}^{\mathrm{FC}}[i] = f_{\mathrm{b}}^{\ast}[\langle-i\rangle_{N}]$, where $\langle i\rangle_{N}$ is the modulo-$N$ remainder of $i$. We define $\mathcal{P}_{\Omega}(\cdot)$ as the subsampling operator that selects the entries of a vector at the indices in $\Omega$. The CS measurements at the RX can then be expressed as the subsampled circular convolution of $\mathbf{f}_{\mathrm{b}}^{\mathrm{FC}}$ and $\mathbf{h}$ \cite{myers2019falp}, i.e.,
\begin{equation}\label{eqn:III1}
    \mathbf{y} = \mathcal{P}_{\Omega}\left(\mathbf{h}\circledast\mathbf{f}_{\mathrm{b}}^{\mathrm{FC}}\right) + \mathbf{n}.
\end{equation}
Due to the circulant structure of the beam training vectors in convolutional CS, $\{\mathbf{f}_m\}^{M-1}_{m=0}$ have the same DFT magnitude as $\mathbf{f}_{\mathrm{b}}$ \cite{myers2019falp}. As a result, all these vectors ``focus'' energy within the desired sector of interest, when $\mathbf{f}_{\mathrm{b}}$ is chosen after SLS.
\par We now discuss how convolutional CS can be interpreted as partial DFT CS \cite{myers2019falp}. We define the masked beamspace channel as $\mathbf{x} = \mathbf{U}_{N}^{\ast}(\mathbf{h}\circledast \mathbf{f}_{\mathrm{b}}^{\mathrm{FC}})$, i.e., the inverse DFT of $\mathbf{h}\circledast \mathbf{f}_{\mathrm{b}}^{\mathrm{FC}}$. By the circular convolution property of the DFT \cite{manolakis2011applied},
\begin{equation}
\label{eq:masked_beam}
\mathbf{x} =[\mathbf{U}_{N}^{\ast}\mathbf{h}]\odot [ \sqrt{N} \mathbf{U}_{N}^{\ast}\mathbf{f}_{\mathrm{b}}^{\mathrm{FC}}].
\end{equation}
 The angle domain spectral mask associated with the masked beamspace is defined as
 \begin{equation}\label{eqn:Ang_spctrl_msk}
\mathbf{p}=\sqrt{N} \mathbf{U}_{N}^{\ast}\mathbf{f}_{\mathrm{b}}^{\mathrm{FC}}.
 \end{equation}
 We note that $\mathbf{x}$ is sparse as it is an element-wise product of the sparse beamspace channel $\mathbf{g}=\mathbf{U}_{N}^{\ast}\mathbf{h}$ and the spectral mask $\mathbf{p}$. Since $\mathbf{h}\circledast \mathbf{f}_{\mathrm{b}}^{\mathrm{FC}}= \mathbf{U}_{N} \mathbf{x}$, we can rewrite \eqref{eqn:III1} as
\begin{equation}\label{eqn:samplingOperator}
\mathbf{y} = \mathcal{P}_{\Omega}\left(\mathbf{U}_{N}\mathbf{x}\right) + \mathbf{n}.
\end{equation}
The CS measurements at the RX are thus a noisy subsampled version of the DFT of the sparse masked beamspace.
\vspace{-1mm}
\subsection{In-sector convolutional CS}
\par Now, we describe the structure of the spectral mask $\mathbf{p}$ used for the in-sector CS problem. We define the sector of interest as $\mathcal{L}=\{d_1,d_1+1,d_1+2, \cdots, d_2\}$, for some integers $d_1<N$ and $d_2<N$. Note that $d_2>d_1$. We assume that the base transmit beam training vector, determined after SLS, results in a spectral mask that has a uniform magnitude at the indices in $\mathcal{L}$. For example, $\mathcal{L}=[N]$ corresponds to a quasi-omnidirectional beam that can be realized using a Zadoff-Chu sequence \cite{myers2019falp}. We define the sector dimension as $N_{\text{sec}}=|\mathcal{L}|$, which is the cardinality $d_2-d_1+1$. Further, we assume that $N_{\text{sec}}$ divides $N$ and define the ratio
\begin{equation}
\label{eq:defn_rho}
\rho=N/N_{\text{sec}}.
\end{equation}
The base transmit beam training vector $\mathbf{f}_{\mathrm{b}}$ is constructed such that the spectral mask $\mathbf{p}$ takes the form
\begin{equation}
\label{eq:des_mask}
    \mathbf{p}=[\underset{d_1 - 1 \, \text{times}}{\underbrace{0, \cdots,0,}}\, e^{\imj \phi_{d_1}}, e^{\imj \phi_{d_1+1}}\cdots,e^{\imj \phi_{d_2}},\, \underset{N - d_2 \, \text{times}}{\underbrace{0, \cdots,0}}]^T,
\end{equation}
to focus only along the directions associated with $\mathcal{L}$. To determine the base beam training vector $\mathbf{f}_{\mathrm{b}}$, \eqref{eqn:Ang_spctrl_msk} is inverted for a random choice of $\{\phi_{d_1},\phi_{d_1+1},\cdots , \phi_{d_2}\}$ in \eqref{eq:des_mask}, i.e.,
$\mathbf{f}_{\mathrm{b}}^{\mathrm{FC}}=\mathbf{U}_{N}\mathbf{p}/ \sqrt{N}$. Then, $\mathbf{f}_{\mathrm{b}}^{\mathrm{FC}}$ is flipped and conjugated to get $\mathbf{f}_{\mathrm{b}}$. Such a simple DFT-based construction for the in-sector beam, equivalently the spectral mask, ignores the resolution of the phase shifters and the amplitude control elements. In future, we will extend our design to low-resolution arrays.
\par When the in-sector beam training vector associated with \eqref{eq:des_mask} is used in convolutional CS, we observe that the sparse masked beamspace vector $\mathbf{x}=\mathbf{g}\odot \mathbf{p}$ is always zero over $[N] \setminus \mathcal{L}$. Here, the interest is in estimating the $N_{\text{sec}}$ coefficients of the sparse masked beamspace $\mathbf{x}$ at the indices in $\mathcal{L}$. Let $\hat{\mathbf{x}}$ be the estimate obtained via CS. Then, the entries of the beamspace estimate $\hat{\mathbf{g}}$ at the indices in $\mathcal{L}$ are computed using
\begin{equation}
     \hat{g}[\ell]=\hat{x}[\ell]p^{\ast}[\ell], \,\forall \ell \in \mathcal{L}.
\end{equation}
Our goal is to design the circulant shifts in $\Omega$ to achieve better reconstruction of $\mathbf{x}$, equivalently $\mathbf{g}$, at the indices in $\mathcal{L}$.
\vspace{-1mm}
\subsection{Circulant shift design}
We design the circulant shifts to minimize the mutual coherence \cite{ben2010coherence} of the CS matrix. A small mutual coherence is desirable as it minimizes the upper bound on the mean squared error (MSE) in the estimated sparse signal \cite{ben2010coherence}.
\par Now, we explicitly write down the CS matrix associated with \eqref{eqn:samplingOperator}. First, we replace the subsampling operator $\mathcal{P}_{\Omega}(\cdot)$ in \eqref{eqn:samplingOperator} by an $M\times N$ subsampling matrix $\mathbf{S}$. The matrix $\mathbf{S}$ is $1$ at $\{(m,c[m])\}^{M-1}_{m=0}$ and is zero at the remaining locations. We define $\mathbf{A} = \mathbf{S}\mathbf{U}_N$ and rewrite \eqref{eqn:samplingOperator} as
\begin{equation}\label{eqn:partialDFT}
    \mathbf{y} = \mathbf{A}\mathbf{x} + \mathbf{n}.
\end{equation}
The sequence of circulant shifts $\{c[m]\}^{M-1}_{m=0}$ determines the matrix $\mathbf{S}$ and the resulting CS matrix $\mathbf{A}$. 

\par As the sparse masked beamspace $\mathbf{x}$ is zero at the indices in $[N] \setminus \mathcal{L}$, the measurement model in \eqref{eqn:partialDFT} can be rewritten as
\begin{equation}\label{eqn:trunc_CS}
    \mathbf{y} = \mathbf{A}_{\mathcal{L}}\mathbf{x}_{\mathcal{L}} + \mathbf{n},
\end{equation}
where $\mathbf{A}_{\mathcal{L}}$ is an $M\times N_\text{sec}$ sub-matrix of $\mathbf{A}$ obtained by retaining only those columns with indices in $\mathcal{L}$, and $\mathbf{x}_{\mathcal{L}}$ is a sub-vector of $\mathbf{x}$ at the indices in $\mathcal{L}$.
For the model in \eqref{eqn:trunc_CS}, the coherence of the effective CS matrix $\mathbf{A}_{\mathcal{L}}$ is
\begin{equation}\label{eqn:III4}
    \mu = \underset{\{(i, j): i\neq j, i\in\mathcal{L}, j\in\mathcal{L}\}}{\max}\ \ |\mathbf{a}_{i}^{\ast}\mathbf{a}_{j}|.
\end{equation}
The coherence $\mu$ is next expressed using the notion of point spread function (PSF) in CS\cite{lustig2007sparse}.
\par We note that $\mu $ in \eqref{eqn:III4} is the maximum off-diagonal entry of an $N_\text{sec}\times N_\text{sec}$ submatrix of $|\mathbf{A}^{\ast} \mathbf{A}|$. To examine the entries of this submatrix, we define the Gram matrix
\begin{equation}
    \mathbf{G}=\mathbf{A}^{\ast} \mathbf{A}.
\end{equation}
We also define an $N\times 1$ binary vector $\mathbf{b}_{\Omega}$ that is $1$ only at the indices in the circulant shift set $\Omega$ and is zero at the other locations. Note that $\sum_{i=1}^{N}b_{\Omega}[i]=M$. We observe that $\mathbf{G}=\mathbf{U}^{\ast}_{N}\mathbf{S}^{\ast}\mathbf{S} \mathbf{U}_{N}$ and $\mathbf{S}^{\ast} \mathbf{S} = \mathrm{Diag}(\mathbf{b}_{\Omega})$. Putting these observations together, we have
\begin{equation}\label{eqn:III6}
\mathbf{G} =  \mathbf{U}_{N}^{\ast}\mathrm{Diag}(\mathbf{b}_{\Omega})\mathbf{U}_{N}.
\end{equation}
Note that $\mathbf{G}$ is circulant as any matrix of the form in \eqref{eqn:III6} is circulant \cite{manolakis2011applied}. The first row of $\mathbf{G}$ \cite{manolakis2011applied}, known as PSF in the literature \cite{lustig2007sparse}, is given by
\begin{equation}
\label{eq:PSF_defn}
    \mathrm{PSF}=\frac{1}{\sqrt{N}}\mathbf{U}_N\mathbf{b}_{\Omega}.
\end{equation}
As $\mathbf{b}_{\Omega}$ is real, the magnitude of the PSF is even symmetric.
\par We leverage the circulant structure of $\mathbf{G}$ and the even symmetry of the PSF to express the coherence in \eqref{eqn:III4} as
\begin{equation}\label{eqn:III7}
    \mu= \underset{i\in\mathcal{I}}{\max}\ \ |\mathrm{PSF}[i]|,
\end{equation}
where $\mathcal{I} = \{1,2,\cdots,(d_2-d_1)\}$.
The coherence minimization problem can then be formulated as
\begin{equation}\label{eqn:opt}
\mathcal{O}:
  \begin{cases}
    \begin{aligned}
        &\underset{\mathbf{b}_{\Omega}\in\{0,1\}^N}{{\min}}
        && \underset{i\in\mathcal{I}}{{\max}} \ \ \ |\mathrm{PSF}[i]|\\
        &\text{{s.t.}} && \sum_{t=1}^{N}b_{\Omega}[t]=M.
    \end{aligned}
  \end{cases}
\end{equation}
As solving the non-convex optimization problem in $\mathcal{O}$ is hard, we find a binary vector $\mathbf{b}^{\mathrm{opt}}$ that achieves $\mu=0$ for $M=N_{\text{sec}}$. Then, we propose a randomized subsampling technique for an arbitrary $M$ that approaches $\mathbf{b}^{\mathrm{opt}}$ as $M \rightarrow N_{\text{sec}}$.

\par To motivate our approach, let us first assume $d_1=0$, $d_2=N-1$, and $M=N$, and solve the problem in $\mathcal{O}$. Here, the optimal solution for $\mathbf{b}_{\Omega}$ is the all-ones vector which corresponds to $\Omega=[N]$. In practice, however, we are interested in the regime where $M \ll N$. When $M \ll  N$, prior work on CS discusses that random subsampling, where $\Omega$ comprises $M$ random samples drawn without replacement from $\{0,1,...,N-1\}$ with uniform probabilities, is a good choice \cite{candes2008introduction}. Such a construction for $\mathbf{b}_{\Omega}$ converges to the optimal all-ones configuration as $M \rightarrow N$.
\par For the in-sector approach, we observe from \eqref{eqn:III7} that not all entries of the PSF need to be zero to achieve the optimal solution, i.e., $\mu = 0$. Specifically, only those entries of the PSF that belong to the set $\mathcal{I}$ have to be zero. In Lemma \ref{lemma}, we discuss the optimal solution when $M=N_{\text{sec}}$.
\begin{lemma}\label{lemma}
For $M = N_{\text{sec}}$, $\mathbf{b}_{\Omega}$ corresponding to $\Omega = \left\lbrace 0, \rho, 2\rho, \cdots  (N_{\text{sec}}-1)\rho \right\rbrace$, with $\rho$ defined in \eqref{eq:defn_rho}, is an optimal solution of $\mathcal{O}$ as it achieves $\mu=0$. This uniform pattern in $\Omega$ was also discussed in \cite{chen2020convolutional} for CS via linear convolution.
\end{lemma}
\begin{proof}
We denote an all-ones vector of length $N_{\text{sec}}$ by $\mathbf{1}_{N_{\text{sec}}}$. Then, we can express $\mathbf{b}_{\Omega}$ as the upsampled version of $\mathbf{1}_{N_{\text{sec}}}$ by $\rho$ \cite[Section~7.4.7]{manolakis2011applied}. This is obtained by inserting $\rho-1$ zeros between consecutive entries of $\mathbf{1}_{N_{\text{sec}}}$. Therefore, the upsampled vector is $1$ only at the locations in $\Omega=\left\lbrace 0, \rho, 2\rho, \cdots  (N_{\text{sec}}-1)\rho \right\rbrace$. Next, we observe that $\mathbf{U}_{N_{\text{sec}}}\mathbf{1}_{N_{\text{sec}}} = \sqrt{N_{\text{sec}}}\mathbf{e}_{0}$, where $\mathbf{e}_{0}$ is the first column of the $N_{\text{sec}} \times N_{\text{sec}}$ identity matrix. Using the properties of the DFT for upsampled sequences over \eqref{eq:PSF_defn}\cite[Section 7.4.7]{manolakis2011applied}, we can write
\begin{equation}\label{eqn:psf_final}
    \mathrm{PSF}[i] = {\rho}^{-1}e_{0}[\langle i\rangle_{N_{\text{sec}}}].
\end{equation}
We observe from \eqref{eqn:psf_final} that $\mathrm{PSF}[i] = 0$ for $i\in\mathcal{I}$. Further, the PSF repeats every $N_{\text{sec}}$ units, i.e., $\mathrm{PSF}[i+k N_{\text{sec}}]=\mathrm{PSF}[i]\, \forall i,k$. This periodic structure is shown in Fig. \ref{fig:uc1A}.
\end{proof}
\par In Fig. \ref{fig:uc1A}, we show the PSFs associated with the subsampling scheme in Lemma \ref{lemma} and a particular realization of random subsampling for the same number of measurements, i.e., $M=N_{\text{sec}}$. We refer to the approach in Lemma \ref{lemma} as proposed circular shifts (PCS), and to the fully random subsampling method as random circular shifts (RCS). We observe that the PSF is $0$ within the region of interest, i.e., $\mathcal{I}$, for PCS in Lemma \ref{lemma}. Also, we notice that RCS which selects $M$ indices at random from $[N]$ results in a larger $\mu$.
\par When the number of measurements is less than the sector dimension, i.e., $M<N_{\text{sec}}$, we propose to select $M$ entries from $\left\lbrace 0, \rho, 2\rho, \cdots  (N_{\text{sec}}-1)\rho \right\rbrace$ uniformly at random. Specifically, $\Omega$ in PCS for $M<N_{\text{sec}}$ is a subset of the optimal set in Lemma \ref{lemma} for $M=N_{\text{sec}}$. Thus, PCS converges to the optimal solution in Lemma \ref{lemma} as $M \rightarrow N_{\text{sec}}$. In Fig.~\ref{fig:uc1B}, we plot the cumulative distribution function (CDF) of $\mu$ with the proposed scheme. We observe that PCS achieves a smaller $\mu$ than the RCS. Fig.~\ref{fig:prpsd} shows a specific realization of PCS and RCS for in-sector CS. Finally, the proposed in-sector CS methods have the same computational complexity as standard CS. This is because we just change the CS matrix to a new randomized construction, and make no changes to the CS algorithm.
\vspace{-6mm}
\begin{figure}[!h]
\centering
\subfloat{\includegraphics[scale=1.1]{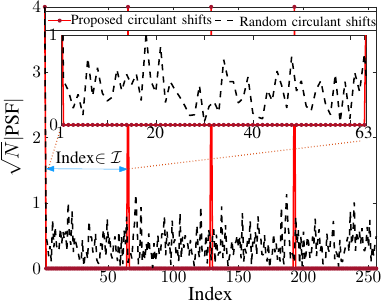}\label{fig:uc1A}}
\hfil
\subfloat{\includegraphics[scale=1.1]{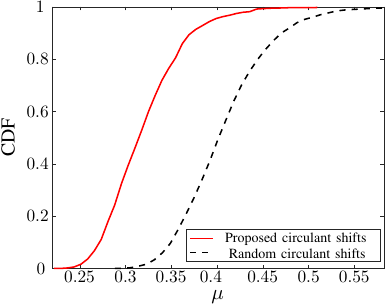}\label{fig:uc1B}}
\vspace{-2mm}
\caption{\small Proposed subsampling (PCS) and fully random subsampling (RCS) for $N=256$ and $N_{\text{sec}}=64$. (a) PSFs for $M=N_{\text{sec}}$ indicate that PCS achieves $\mu=0$ over $\mathcal{I}$. (b) For $M=25$, we observe from the CDF that PCS achieves a smaller $\mu$ than the fully random RCS.
\normalsize}
\end{figure}
\vspace{-2mm}
\begin{figure}[!h]
\vspace{-01mm}
\centering
\includegraphics[scale=1]{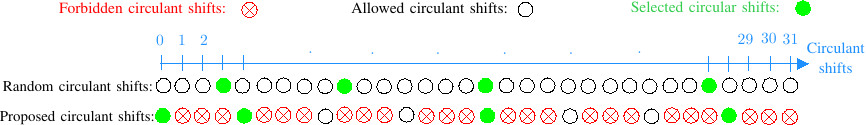}
\vspace{-02mm}
\caption{\small Proposed circulant shifts (PCS) for $N=32$, $N_{\text{sec}}=8$, and $M=4$. Here, $\rho=32/8=4$. The $M$ circulant shifts in PCS and RCS are chosen at random from $\{0,4,8,\cdots 28\}$ and $\{0,1,2,\cdots 31\}$ respectively. In this example, $\Omega=\{0,4,16,28\}$ with PCS.
\normalsize}\label{fig:prpsd}
\vspace{-4mm}
\end{figure}
\section{Simulations}\label{sec5}
We consider a narrowband mmWave system operating at a carrier frequency of $38$ GHz, with $N=256$ antennas at the TX and a single antenna RX. The TX-RX distance is $15$ $\mathrm{m}$. We use $100$ urban micro non-line-of-sight channels from the NYU simulator \cite{sun2017novel}. Note that the AoDs can be off-grid. The beamspace is divided into $4$ sectors of equal dimension, i.e., $N_{\text{sec}}=64$ and thus $\rho=4$. The SNR in the received CS measurements is $\text{SNR}= \mathbb{E}[\Vert \mathbf{Fh}\Vert^2]/(M\sigma^2)$. 
\par We now explain SLS and in-sector channel estimation in our simulations. First, the TX applies each of the $4$ sectored beam patterns, and the RX records the received power with each beam pattern. Next, the sector that results in the largest received power is selected. Then, in-sector channel estimation is performed with different CS solutions using the orthogonal matching pursuit (OMP) algorithm \cite{tropp2007signal} to recover beamspace entries within the desired sector, i.e., $\hat{\mathbf{g}}_\mathcal{L}$. The channel estimate is then  $\hat{\mathbf{h}}=[\mathbf{U}_{N}]_{\mathcal{L}}\hat{\mathbf{g}}_\mathcal{L}$. The normalized MSE (NMSE) within the desired sector is defined as $\mathbb{E}[\|\mathbf{g}_{\mathcal{L}}-\mathbf{\hat{g}}_{\mathcal{L}}\|^2]/\mathbb{E}[\|\mathbf{g}_{\mathcal{L}}\|^2]$. The achievable rate is defined as $\log_{2}(1+|\mathbf{f}^{\text{mrt}}\mathbf{h}|^2/\sigma^2)$, where $\mathbf{f}^{\text{mrt}} = \mathbf{\hat{h}}^{\ast}/\|\mathbf{\hat{h}}\|$. Although this work estimates the channel only within the desired sector, our algorithm can be applied independently over each sector to estimate the entire angle domain channel. Such an approach is robust to noise when compared to standard CS techniques that use wide beams.
\par We also evaluate performance with the proposed CS beams when a $4\times$ oversampled DFT dictionary (say $\mathbf{D}$ of size $N \times 4N$) \cite{duarte2013spectral} is used to represent the channel, i.e., $\mathbf{h}=\mathbf{D}\tilde{\mathbf{g}}$. In this case, the same design of circulant beam training sequences is used in \eqref{eqn:RXM} to solve for $\tilde{\mathbf{g}} \in \mathbb{C}^{4N \times 1}$. Here, the $M \times 4N$ CS matrix $\mathbf{FD}$ is used in OMP. Then, a $4N$ dimensional rectangular window that is $1$ within the sector of interest is applied in every iteration of OMP to get $\hat{\tilde{\mathbf{g}}}_{\mathrm{win}}$. Finally, the in-sector channel estimate is computed as $\hat{\mathbf{h}}=\mathbf{D}\hat{\tilde{\mathbf{g}}}_{\mathrm{win}}$.
\begin{figure}[!t]
\centering
\includegraphics[trim=0.2cm 0cm 0.2cm 0cm, width=0.27\textwidth]{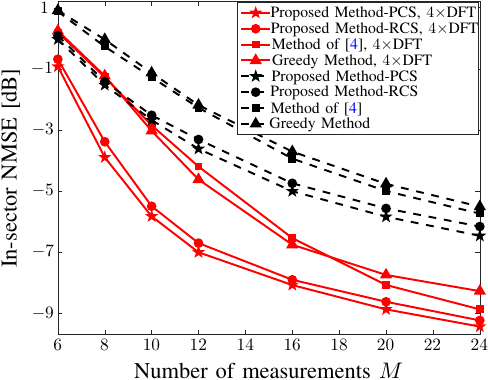}
\vspace{-02mm}
\caption{\small In-sector NMSE with the number of CS measurements for $N=256$, $N_{\text{sec}}=64$, and $\text{SNR}=5$ dB. PCS results in a smaller $\mu$ and achieves a lower NMSE than the fully random RCS approach.
\normalsize}\label{fig:nr1}
\end{figure}
\begin{figure}[!t]
\vspace{-02mm}
\centering
\includegraphics[trim=0.2cm 0cm 0.2cm 0cm, width=0.27\textwidth]{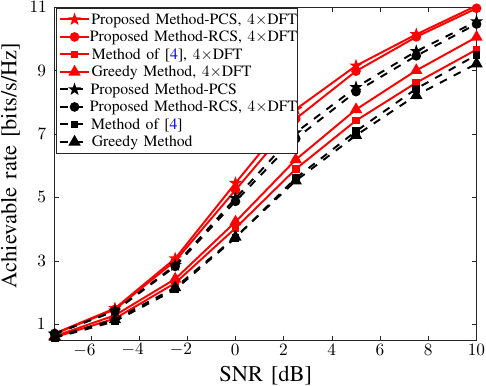}
\vspace{-03mm}
\caption{\small The achievable rate with our method is larger than that with the benchmarks. Here, $N=256$, $N_{\text{sec}}=64$, and $M=10$.
\normalsize}\label{fig:nr2}
\vspace{-4mm}
\end{figure}
\par To construct a base beam training vector $\mathbf{f}_{\mathrm{b}}$ for our approach, we generate $5000$ beamspace candidates for $\mathbf{p}$ that are all consistent with \eqref{eq:des_mask}. This is done by choosing $\{\phi\}^{d_2}_{d_1}$ uniformly at random. Then, the candidate that results in the lowest peak to average power ratio for $\mathbf{f}_{\mathrm{b}}$ is chosen to realize the beamformer using amplitude control elements of a limited dynamic range. For a fair comparison with \cite{tsai2018structured}, the same $\mathbf{f}_{\mathrm{b}}$ is used in the method of \cite{tsai2018structured}. For the greedy method \cite{ali2017millimeter}, $30M$ random unit-norm complex vectors are first generated. Then, the top $M$ vectors that radiate the largest energy within the sector of interest are used to acquire CS measurements.
\par We study the NMSE versus $M$ and the achievable rate versus SNR. Fig.~\ref{fig:nr1} shows the NMSE with our convolutional CS-based in-sector channel estimation, when circulant shifts are chosen according to PCS and RCS. The same beamformer was used to evaluate PCS and RCS in our method. PCS results in a lower NMSE than RCS because the former results in lower aliasing artifacts within the sector of interest, equivalently a smaller $\mu$, as observed in Fig. \ref{fig:uc1B}. Our method also outperforms other benchmarks in \cite{ali2017millimeter} and \cite{tsai2018structured} that result in an energy leakage outside the sector of interest.
Finally, we observe from Fig.~\ref{fig:nr2} that our approach achieves a slightly higher rate with proposed circulant shifts than random shifts, for the same computational complexity. The rate is also higher than the other benchmarks.
\section{Conclusions}\label{sec6conclusion}
In this paper, we adopt a convolutional CS framework to design CS matrices that are well suited for beam training within a sector of interest. The compressive beam training vectors in our approach always focus on the sector of interest, resulting in a higher SNR in the CS channel measurements. We also designed a new randomized set of circulant shifts in convolutional CS for spatial channel estimation within the sector of interest. The designed set of shifts achieves better in-sector channel reconstruction than fully random shifts.


\ifCLASSOPTIONcaptionsoff
  \newpage
\fi

\bibliography{References}

\begin{thebibliography}{10}

\bibitem{candes2008introduction}
E.~J. Cand{\`e}s and M.~B. Wakin, ``An introduction to compressive sampling,''
  {\em IEEE Signal Process. Mag.}, vol.~25, no.~2, pp.~21--30, 2008.

\bibitem{alkhateeb2015compressed}
A.~Alkhateeb, G.~Leus, and R.~W. Heath, ``Compressed sensing based multi-user
  millimeter wave systems: {How} many measurements are needed?,'' in {\em Proc.
  IEEE Intl. Conf. Acoustics, Speech and Signal Process. (ICASSP)},
  pp.~2909--2913, IEEE, 2015.

\bibitem{myers2019falp}
N.~J. Myers, A.~Mezghani, and R.~W. Heath, ``{FALP}: Fast beam alignment in
  {mmWave} systems with low-resolution phase shifters,'' {\em IEEE Trans.
  Commun.}, vol.~67, no.~12, pp.~8739--8753, 2019.

\bibitem{tsai2018structured}
C.-R. Tsai and A.-Y. Wu, ``Structured random compressed channel sensing for
  millimeter-wave large-scale antenna systems,'' {\em IEEE Trans. Signal
  Process.}, vol.~66, no.~19, pp.~5096--5110, 2018.

\bibitem{wang2021jittering}
W.~Wang and W.~Zhang, ``Jittering effects analysis and beam training design for
  {UAV} millimeter wave communications,'' {\em IEEE Trans. on Wireless
  Commun.}, 2021.

\bibitem{heath2016overview}
R.~W. Heath, N.~Gonzalez-Prelcic, S.~Rangan, W.~Roh, and A.~M. Sayeed, ``An
  overview of signal processing techniques for millimeter wave {MIMO}
  systems,'' {\em IEEE J. Sel. Topics Signal Process.}, vol.~10, no.~3,
  pp.~436--453, 2016.

\bibitem{ali2017millimeter}
A.~Ali, N.~Gonz{\'a}lez-Prelcic, and R.~W. Heath, ``Millimeter wave
  beam-selection using out-of-band spatial information,'' {\em IEEE Trans. on
  Wireless Commun.}, vol.~17, no.~2, pp.~1038--1052, 2017.

\bibitem{chen2020convolutional}
P.-C. Chen and P.~Vaidyanathan, ``Convolutional beamspace for linear arrays,''
  {\em IEEE Trans. Signal Process.}, vol.~68, pp.~5395--5410, 2020.

\bibitem{li2012convolutional}
K.~Li, L.~Gan, and C.~Ling, ``Convolutional compressed sensing using
  deterministic sequences,'' {\em IEEE Trans. Signal Process.}, vol.~61, no.~3,
  pp.~740--752, 2012.

\bibitem{nitsche2014ieee}
T.~Nitsche, C.~Cordeiro, A.~B. Flores, E.~W. Knightly, E.~Perahia, and J.~C.
  Widmer, ``{IEEE} 802.11 ad: directional 60 {GHz} communication for
  multi-gigabit-per-second {Wi-Fi},'' {\em IEEE Commun. Mag.}, vol.~52, no.~12,
  pp.~132--141, 2014.

\bibitem{manolakis2011applied}
D.~G. Manolakis and V.~K. Ingle, {\em Applied digital signal processing: theory
  and practice}.
\newblock Cambridge university press, 2011.

\bibitem{ben2010coherence}
Z.~Ben-Haim, Y.~C. Eldar, and M.~Elad, ``Coherence-based performance guarantees
  for estimating a sparse vector under random noise,'' {\em IEEE Trans. Signal
  Process.}, vol.~58, no.~10, pp.~5030--5043, 2010.

\bibitem{lustig2007sparse}
M.~Lustig, D.~Donoho, and J.~M. Pauly, ``Sparse {MRI}: The application of
  compressed sensing for rapid {MR} imaging,'' {\em Magnetic Resonance in
  Medicine}, vol.~58, no.~6, pp.~1182--1195, 2007.

\bibitem{sun2017novel}
S.~Sun, G.~R. MacCartney, and T.~S. Rappaport, ``A novel millimeter-wave
  channel simulator and applications for {5G} wireless communications,'' in
  {\em Proc. of IEEE Intl. Conf. on Commun. (ICC)}, pp.~1--7, 2017.

\bibitem{tropp2007signal}
J.~A. Tropp and A.~C. Gilbert, ``Signal recovery from random measurements via
  orthogonal matching pursuit,'' {\em IEEE Trans. on inform. theory}, vol.~53,
  no.~12, pp.~4655--4666, 2007.

\bibitem{duarte2013spectral}
M.~F. Duarte and R.~G. Baraniuk, ``Spectral compressive sensing,'' {\em
  Elsevier Appl. and Comput. Harmonic Analysis}, vol.~35, no.~1, pp.~111--129,
  2013.

\end{thebibliography}
\bibliographystyle{ieeetr}

\end{document}